\documentclass[preprint,authoryear,times,12pt]{elsarticle}

\usepackage{amsfonts, amsmath, amssymb, bm,  graphicx, calrsfs, color}
\usepackage{appendix}
\usepackage{mathrsfs}
\usepackage{paralist}
\usepackage{booktabs}

\newtheorem{assump}{Assumption}{\bf}{\it}
\newtheorem{theorem}{Theorem}
\newtheorem{proposition}{Proposition}
\newtheorem{lemma}{Lemma}
\newproof{proof}{Proof}

\makeatletter

\newcommand{\Rmnum}[1]{\expandafter\@slowromancap\romannumeral #1@}
\makeatother

\newcommand{\mv}{\bm}

\renewcommand{\le}{\leqslant}
\renewcommand{\ge}{\geqslant}

\DeclareMathOperator{\proba}{\mathrm{P}}
\DeclareMathOperator{\proban}{\mathbb{P}_{n}}

\newcommand{\dto}{\rightsquigarrow}

\newcommand{\RR}{\mathbb{R}}
\newcommand{\simplexp}{\Delta_{p-1}}
\newcommand{\simplextwo}{\Delta_{2}}
\newcommand{\cont}{{\cal C}}

\newcommand{\Aclass}{\mathcal{A}}
\newcommand{\Aclassm}{\mathcal{A}_{m}}

\newcommand{\inv}{^\leftarrow}

\newcommand{\defni}{\emph}
\newcommand{\1}{\boldsymbol{1}}
\newcommand{\diff}{\mathrm{d}}

\newcommand{\ecopn}{\mathbb{C}_n}

\newcommand{\An}{\hat{A}_{n} }

\newcommand{\ACFGn}{\hat{A}_{n}^{\mathrm{CFG}}}
\newcommand{\APn}{\hat{A}_{n}^{\mathrm{P}}}

\newcommand{\aAP}{\mathbb{A}^{\mathrm{P}}}
\newcommand{\aACFG}{\mathbb{A}^{\mathrm{CFG}}}

\newcommand{\aACFGn}{\mathbb{A}^{\mathrm{CFG}}_n}
\newcommand{\aAPn}{\mathbb{A}^{\mathrm{P}}_n}

\newcommand{\Aproj}{\hat{A}^{\mathrm{pr}}}

\newcommand{\Aprojm}{ \hat{A}^{\mathrm{pr}}_{m}}

\newcommand{\eps}{\varepsilon}
\renewcommand{\Pr}{\operatorname{P}}
\newcommand{\E}{\operatorname{E}}
\newcommand{\cov}{\operatorname{cov}}

\newcommand{\argmin}{\operatornamewithlimits{\arg\min}}

\newcommand{\Hc}{\mathscr{H}}
\newcommand{\Vc}{\mathscr{V}}

\newcommand{\inpr}[2]{\langle{#1},{#2}\rangle}
\newcommand{\norm}[1]{\|{#1}\|}

\newcommand{\ind}{\bm{1}}

\numberwithin{equation}{section}

\begin{document}

\begin{frontmatter}
\title{Nonparametric estimation of multivariate extreme-value copulas} 

\author[UCL]{Gordon Gudendorf\fnref{fn1}}
\ead{gordon.gudendorf@uclouvain.be}

\author[UCL]{Johan Segers\fnref{fn1}\corref{cor}} 
\ead{johan.segers@uclouvain.be}

\address[UCL]{Institut de statistique, biostatistique et sciences actuarielles, Voie du Roman Pays 20, B-1348 Louvain-la-Neuve, Belgium.}
\fntext[fn1]{Funding was provided by IAP research network grant P6/03 of the Belgian government (Belgian Science Policy) and by ``Pro\-jet d'ac\-tions de re\-cher\-che con\-cer\-t\'ees'' number 07/12-002 of the Com\-mu\-nau\-t\'e fran\-\c{c}ai\-se de Bel\-gi\-que, granted by the Aca\-d\'e\-mie uni\-ver\-si\-taire de Lou\-vain.}
\cortext[cor]{Corresponding author}

\begin{abstract}
Extreme-value copulas arise in the asymptotic theory for componentwise maxima of independent random samples. An extreme-value copula is determined by its Pickands dependence function, which is a function on the unit simplex subject to certain shape constraints that arise from an integral transform of an underlying measure called spectral measure. Multivariate extensions are provided of certain rank-based nonparametric estimators of the Pickands dependence function. The shape constraint that the estimator should itself be a Pickands dependence function is enforced by replacing an initial estimator by its best least-squares approximation in the set of Pickands dependence functions having a discrete spectral measure supported on a sufficiently fine grid. Weak convergence of the standardized estimators is demonstrated and the finite-sample performance of the estimators is investigated by means of a simulation experiment.
\end{abstract}

\end{frontmatter}

\section{Introduction}
\label{s:intro}

Extreme-value copulas arise in the asymptotic theory for componentwise maxima of independent random samples. They provide the dependence structures for the class of multivariate extreme-value or max-stable distributions. More generally, they constitute a flexible class of models for describing positive association; see \citet{GdS10} for a survey.

In this paper we will focus on the nonparametric estimation of extreme-value copulas in general dimensions. In particular, we aim at multivariate extensions of the rank-based estimators in \citet{GS09} and the projection methodology in \citet{FGS08}.

Let $\mv{X}_i = (X_{i,1}, \ldots, X_{i,p})$, $i \in \{1, \ldots, n\}$, be an independent random sample from a $p$-variate, continuous distribution function $F$ with margins  $F_1, \ldots, F_p$ and copula $C$, that is,
\[
  F(\mv{x}) = C \bigl( F_1(x_1), \ldots, F_p(x_p) \bigr), \qquad \mv{x} \in \RR^p,
\]
where $F(\mv{x}) = \Pr(\mv{X} \le \mv{x})$ (componentwise inequalities), $F_j(x_j) = \Pr(X_j \le x_j)$, and $C$ is the joint distribution function of $(F_1(X_1), \ldots, F_p(X_p))$. We are interested in nonparametric estimation of $C$ in the model where the margins $F_1, \ldots, F_p$ are completely unknown (but continuous) and $C$ is known to be an extreme-value copula.

A $p$-variate copula $C$ is an \defni{extreme-value copula} if there exists a finite Borel measure $H$ on the unit simplex $\simplexp = \{ (w_{1} , \dots , w_{p}) \in [0 , 1]^p : w_1 + \cdots + w_p = 1 \}$, called \defni{spectral measure}, such that
\begin{equation}
\label{MEVC:ell}
  C( \mv{u} ) = \exp \bigl( - \ell( - \log u_1, \ldots, - \log u_{p} ) \bigr), \qquad \mv{u} \in (0, 1]^p,
\end{equation}
the \defni{tail dependence function} $\ell : [0, \infty)^p \to [0, \infty)$ being given by
\begin{equation}
\label{Ell}
  \ell(\mv{x}) = \int_{\simplexp} \max\{v_1 x_1, \ldots, v_p x_p\} \; H(\diff \mv{v} ).
\end{equation}
The spectral measure $H$ is arbitrary except for the $p$ moment constraints
\begin{equation}
  \int_{\simplexp} v_j \; H( \diff \mv{v} ) = 1, \qquad j \in \{1, \ldots, p \},
\label{H:constr}
\end{equation}
which are equivalent to requiring that the margins of $C$ be uniform on $(0, 1)$. 

The tail dependence function $\ell$ in \eqref{Ell} is convex, homogeneous of order one, that is $\ell( c \mv{x} ) = c \, \ell( \mv{x} )$ for $c > 0$, and satisfies $\max(x_1, \ldots, x_p) \le \ell(\mv{x}) \le x_1 + \cdots + x_p$ for all $\mv{x} \in [0, \infty)^p$. By homogeneity, $\ell$ is characterized by the \defni{Pickands dependence function} $A : \simplexp \to [1/p, 1]$, which is simply the restriction of $\ell$ to the unit  simplex: for $\mv{x} \in [0, \infty)^p \setminus \{ \mv{0} \}$,
\begin{align*}
  \ell(\mv{x}) &= (x_1 + \cdots + x_p ) \, A(w_1, \ldots, w_{p-1} ) \\
  \text{where } \qquad w_j &= \frac{x_j}{x_1 + \cdots + x_p }.
\end{align*}
Here and further on, we frequently identify $\simplexp$ with $\{ (w_1, \ldots, w_{p-1}) \in [0, 1]^{p-1} : w_1 + \cdots + w_{p-1} \le 1 \}$. The extreme-value copula $C$ can be expressed in terms of $A$ via
\begin{equation}
  C( \mv{u} )
  = \exp \biggl\{ \bigl( {\textstyle\sum_{j=1}^{p} \log u_{j}} \bigr) \,
    A \biggl(  \frac{\log u_{1}}{ \sum_{j=1}^{p} \log u_{j} }, \ldots ,  \frac{\log u_{p-1}}{ \sum_{j=1}^{p} \log u_{j} }  \biggr) \biggr\} 
 \label{Pick:rep}   
\end{equation}
for $\mv{u} \in (0, 1]^p \setminus \{ (1, \ldots, 1) \}$, with 
\begin{equation}
  A(\mv{w}) =  \int_{\simplexp} \max \{ v_1 w_1, \ldots, v_p w_p \} H(\diff \mv{v} ) , \qquad \mv{w} \in \simplexp,
\label{A:h} 
\end{equation}
see \citet{Pickands81} and \citet{ZWP08}. The function $A$ is convex as well and satisfies $\max(w_1, \ldots, w_p ) \le A( \mv{w} ) \le 1$ for all $\mv{w} \in \simplexp$. 

Nonparametric estimators for $A$ have initially been developed in \citet{Pickands81}, with modifications in \citet{Deheuvels91} and \citet{HT00}, and in \citet{CFG97}. These estimators will be referred to as the Pickands and CFG estimators, respectively; see Section~\ref{s:estim} for definitions. In the previously cited papers, the marginal distributions were assumed to be known. The more realistic case of unknown margins has been treated in the bivariate case in \citet{RVF01} for a submodel and in \citet{GS09} for the general model. Multivariate extensions have been proposed in \citet{ZWP08} and \citet{GdS11} for the case of known margins. In Section~\ref{s:estim}, we will provide a proof for the convergence of these estimators in case of unknown margins being estimated by the empirical distribution functions, thus generalizing the main results in \citet{GS09} to arbitrary dimensions. As in \citet{KY10} and \citet{GKNY11}, the estimators could also be used as a starting point for goodness-of-fit tests, but for brevity we do not pursue this here. Finally, a new type of nonparametric estimator has been proposed in \citet{BDV11} for the bivariate case.

In the proofs of the asymptotic normality of the Pickands and CFG estimators, a certain expansion of the empirical copula process due to \citet{Stute84} and \citet{T05} plays a crucial role. The second-order derivatives of extreme-value copulas typically exhibit explosive behaviour near the corners of the hypercube, violating the assumptions in the two papers just cited. In \citet{Segers11}, it was shown that the same expansion continues to hold under much weaker conditions on the partial derivatives. In Section~\ref{s:empcop}, these issues are considered for multivariate extreme-value copulas.

As the estimators for $A$ considered here fail to be Pickands dependence functions themselves, it is natural to ask how to enforce the shape constraints on such functions in the estimation procedure. In dimension $p = 2$, it is sufficient to ensure that $A$ is convex and takes values in the range $\max(\mv{w}) \le A(\mv{w}) \le 1$, for instance by truncation and convexification \citep{Deheuvels91}. In dimension $p \ge 3$, however, this procedure is no longer sufficient \citep[page~257]{BGTS04} and one needs to rely on the spectral representation in \eqref{A:h}. In Section~\ref{s:proj} we will apply an projection methodology \citep{FGS08} to obtain valid estimates: an initial estimate is replaced by its best least-squares approximation in the set of Pickands dependence functions corresponding to discrete spectral measures supported on a fine grid.

The results of a simulation experiment aimed at investigating the finite-sample performance of the original and projected Pickands and CFG estimators are reported in Section~\ref{s:simul}. All proofs are relegated to the Appendix. 

Throughout the article, we will apply the following notations. For a space $\mathcal{W}$, let $ \ell^{\infty}( \mathcal{W}  )$ and $\cont ( \mathcal{W} )$ denote the spaces of real-valued bounded and real-valued  continuous functions respectively, where we endow both spaces with the uniform norm $\| \cdot \|_{\infty} : f \mapsto \sup_{x \in \mathcal{W}} | f(x) | $.
Furthermore, the indicator function of the event $E$ is denoted by $\1 (E)$. The arrow `$\dto$' will stand for  weak convergence. 
For any $p-$variate real-valued function $f$ with values in $\RR$, the first and second-order partial derivatives will be denoted by $\dot{f}_{i}(\mv{x}) = \frac{\partial}{\partial x_{i}} f(x_{1}, \dots, x_{p})$ and $\ddot{f}_{ij}(\mv{x}) = \frac{\partial^2}{\partial x_{i} \partial x_{j}} f(x_{1}, \dots, x_{p})$.

\section{Empirical copula processes}
\label{s:empcop}

Let $\mv{X}_1, \mv{X}_2, \ldots$ be an iid sequence of random vectors from a $p$-variate multivariate distribution $F$ with continuous margins $F_1, \ldots, F_p$. If the margins $F_{1}, \dots , F_{p}$ are known, we can define the empirical cumulative distribution function $C_n$ of the (unobservable) random sample $\mv{U}_{i} = (U_{i,1}, \ldots, U_{i,p}) = (F_1(X_{i,1}), \ldots, F_p(X_{i,p}))$ for $i \in \{1, \ldots, n\}$ by
\begin{equation}
\label{e:Cn}
  C_{n}( \mv{u} ) = \frac{1}{n} \sum_{i=1}^{n} \ind\left( U_{i,1} \le u_{1} , \dots ,  U_{i,p} \le u_{p} \right), 
  \qquad \mv{u} \in [0, 1]^p,
\end{equation}
with associated empirical process
\begin{equation}
\label{e:empproc}
\alpha_{n} = n^{1/2} \left( C_{n} - C \right).
\end{equation}
For ease of notation, we will write
\begin{equation}
\alpha_{n,j}(u_j) = \alpha_n (1, \ldots , 1 , u_j , 1 , \ldots ,1)  \quad \text{for $j \in \{ 1, \ldots, p \}$.}
\end{equation}
In practice, the marginal distributions will need to be estimated. If we are not willing to make any assumptions (except for continuity) we can estimate them by the empirical distribution functions
\begin{equation}
\label{e:empdist}
F_{n,j}(x) = \frac{1}{n+1} \sum_{i=1}^{n} \ind(X_{i,j} \le x ), \qquad j\in \{1, \ldots ,p \},
\end{equation} 
where we divided by $n+1$ in order to avoid later problems at the borders. By so doing, we can construct $n$ vectors $\hat{\mv{U}}_i=( \hat{U}_{i,1}, \ldots , \hat{U}_{i,p} )$ via
\begin{equation}
\label{e:ranks}
\hat{U}_{i,j} = F_{n,j}(X_{i,j}) = \frac{1}{n+1} \sum_{l=1}^{n} \ind( X_{l,j} \le X_{i,j})
\end{equation}
for $i \in \{1, \ldots, n\}$ and $j \in \{1, \ldots, p\}$. The empirical copula will be denoted by
\begin{equation}
\label{e:empcop}
\hat{C}_{n}( \mv{u} )  =  \frac{1}{n} \sum_{i=1}^{n} \ind\left( \hat{U}_{i,1} \le u_{1} , \dots ,  \hat{U}_{i,p} \le u_{p} \right),
  \qquad \mv{u} \in [0, 1]^p
\end{equation}
with associated empirical copula process
\begin{equation}
\label{e:empcopproc}
\ecopn = n^{1/2} \bigl( \hat{C}_{n} - C \bigr).
\end{equation}

In \citet{Stute84} and \citet{T05} it was established that if all second-order derivatives of $C$ exist and are continuous on $[0, 1]^p$, the processes $\alpha_n$ in \eqref{e:empproc} and $\ecopn$ in \eqref{e:empcopproc} are related via
\begin{equation}
\label{e:stute}
  \ecopn(\mv{u})  =  \alpha_{n} ( \mv{u} ) - \sum_{j=1}^p \dot{C}_{j}( \mv{u} ) \, \alpha_{n,j}(u_j) + R_{n}( \mv{u} ),
\end{equation}
the remainder term $R_n$ satisfying
\begin{equation}
\label{e:stute:Rn}
  \sup_{\mv{u} \in [0,1]^{p}} \vert R_{n} (\mv{u}) \vert = O \bigl( n^{-1/4}  ( \log n)^{1/2} ( \log \log n )^{1/4} \bigr) \quad \text{almost surely}.
\end{equation}

Let $\ell^\infty([0, 1]^p)$ denote the space of bounded real-valued functions on $[0, 1]^p$, equipped with the topology of uniform convergence. Weak convergence of maps taking values in $\ell^\infty([0, 1]^p)$ will be understood as in \citet[Definition~1.3.3]{VW96} and will be denoted by `$\dto$'. By classical empirical process theory, we have $\alpha_n \dto \alpha$ as $n \to \infty$, the limiting process being a mean-zero tight Gaussian process with continuous trajectories and covariance function
\begin{equation}
\label{e:alpha}
  \cov \bigl( \alpha(\mv{u}), \alpha(\mv{v}) \bigr)
  = C( \mv{u} \wedge \mv{v} ) - C(\mv{u}) \, C(\mv{v}), \qquad \mv{u}, \mv{v} \in [0, 1]^p,
\end{equation}
where $(\mv{u} \wedge \mv{v})_j = \min(u_j, v_j)$. In view of the expansion \eqref{e:stute}, it then follows that in $\ell^\infty([0, 1]^p)$, we have $\mathbb{C}_n \dto \mathbb{C}$ as $n \to \infty$, where
\begin{equation}
\label{e:copproc}
  \mathbb{C}(u) = \alpha(\mv{u}) - \sum_{j=1}^p \dot{C}_j(\mv{u}) \, \alpha_j(u_j)
\end{equation}
and $\alpha_j(u_j) = \alpha(1, \ldots, 1, u_j, 1, \ldots, 1)$.

Like many other copulas, extreme-value copulas do in general not have uniformly bounded second-order partial derivatives. For instance, in the bivariate case, every copula having a positive coefficient of upper tail dependence will have first-order partial derivatives that fail to have a continuous extension to the upper corner $(1, 1)$; see \citet[Example~1.1]{Segers11}. As a consequence, the only bivariate extreme-value copula whose density is uniformly bounded is the independence copula. However, as shown in the same paper, for copulas satisfying Assumption~\ref{Assump:1} below, the expansion~\eqref{e:stute}--\eqref{e:stute:Rn} of the empirical copula process remains valid. 

\begin{assump}
\label{Assump:1}
\begin{itemize}
\item[(C1)] For every $j \in \{1, \ldots, p\}$, the first-order partial derivative $\dot{C}_j$ exists and is continuous on the set $V_{p,j} = \{ u \in [0, 1]^p : 0 < u_j < 1 \}$.
\item[(C2)] For every $i, j \in \{1, \ldots, p\}$ ($i$ and $j$ not necessarily distinct), $\ddot{C}_{ij}$ exists and is continuous on $V_{p,i} \cap V_{p,j}$ and
\[
  \sup_{u \in V_{p,i} \cap V_{p,j}} \max\{ u_i(1-u_i), u_j(1-u_j) \} \, | \ddot{C}_{ij}(u) | < \infty.
\]
\end{itemize}
\end{assump}

In fact, for weak convergence $\mathbb{C}_n \dto \mathbb{C}$ in $\ell^\infty([0, 1]^p)$ to hold, condition (C1) is already sufficient. In the context of multivariate extreme-value copulas, it will be of interest to have a readily verifiable condition on the stable tail dependence function $\ell$ for Assumption~\ref{Assump:1} to hold.

\begin{assump}
\label{Assump:2}
\begin{itemize}
\item[(L1)] For every $j \in \{1, \ldots, p\}$, the first-order partial derivative $\dot{\ell}_j$ exists and is continuous on the set $W_{p,j} = \{ x \in [0, \infty)^p : x_j > 0 \}$.
\item[(L2)] For every $i, j \in \{1, \ldots, p\}$ ($i$ and $j$ not necessarily distinct), $\ddot{\ell}_{ij}$ exists and is continuous on $W_{p,i} \cap W_{p,j}$ and
\[
  \sup_{\substack{x \in W_{p,i} \cap W_{p,j} \\ x_1 + \cdots + x_p = 1}} \max(x_i, x_j) \, |\ddot{\ell}_{ij}(x)| < \infty. 
\]
\end{itemize}
\end{assump}

\begin{proposition}
\label{thm:empcopproc}
For $p$-variate extreme-value copulas, (L1) implies (C1). If in addition (L2) holds, then (C2) holds as well.
\end{proposition}

In the bivariate case, sufficient conditions for (C1) and (C2) can be given in terms of the Pickands dependence function $A(w) = \ell(1-w, w)$, where $w \in [0, 1]$: (C1) holds as soon as $A$ is continuously differentiable on $(0, 1)$, and (C1)--(C2) hold as soon as $A$ is twice continuously differentiable on $(0, 1)$ and $\sup_{0 < w < 1} w(1-w) \, A''(w) < \infty$ \citep[Example~5.3]{Segers11}.

For completeness, we want to mention that in the above references, the empirical copula is not defined as in \eqref{e:empcop} but rather as
\[
  \hat{C}^{D}_{n} ( \mv{u} ) = F_{n} \bigl( F_{n,1}\inv (u_{1}), \dots , F_{n,p}\inv (u_{p}) \bigr),
\] 
with $F_{n,j}\inv(u_j) = \inf \{ x_j \in \RR : F_{n,j}(x_j) \ge u_j \}$. Straightforward calculus shows that, in the absence of ties,
\begin{equation*}
\sup_{\mv{u} \in [0,1]^{p}} \vert \hat{C}^{D}_{n}(\mv{u}) - \hat{C}_{n}(\mv{u})      \vert \le \frac{2p}{n},
 \end{equation*}
As a consequence, Stute's expansion \eqref{e:stute} is valid for $\hat{C}_n$ if and only if it is valid for $\hat{C}_n^D$.

\section{Nonparametric estimation of the dependence function}
\label{s:estim}

Among the most popular nonparametric estimators for $A$ figure the Pickands estimator $\APn$ \citep{Pickands81} and the estimator $\ACFGn$ proposed by \citet{CFG97}, referred to as the CFG estimator from now on. Writing
\[
  \hat{\xi}_{i} (\mv{w})  =  \bigwedge_{j=1}^{p} \frac{- \log \hat{U}_{i,j}}{w_{j}} .
\]
for $\mv{w} \in \simplexp$, with $\hat{\mv{U}}_{i,j}$ as in \eqref{e:ranks}, these estimators are defined as
\[
\frac{1}{\APn( \mv{w} )}  =  \frac{1}{n} \sum_{i=1}^{n} \hat{ \xi}_i (\mv{w}) \qquad \text{and } \quad 
\log \ACFGn( \mv{w} )  =  - \frac{1}{n} \sum_{i=1}^{n} \log \hat{\xi}_i (\mv{w}) - \gamma ,
\]
with $\gamma =  0.5772\ldots$ the Euler--Mascheroni constant. Explanations on the construction of these estimators are provided for instance in the original references given before, in \citet{GS09} and in the survey paper \citet{GdS10}. The multivariate extension of the CFG estimator was presented in \citet{ZWP08}, albeit under a different but equivalent form.   

In order to improve the small-sample properties of the above estimators, the endpoint constraints $A(\mv{e}_j) = 1$ for $j \in \{1, \ldots, p\}$, with $\mv{e}_j$ the $j$th standard unit vector in $\RR^p$, can be imposed as follows. Given continuous functions $\lambda_1, \ldots, \lambda_p : \simplexp \to \RR$ verifying $\lambda_j(\mv{e}_k) = \delta_{jk}$ (Kronecker delta) for all $j,k \in \{1, \ldots, p\}$, define
\begin{align}
\label{eq:pick:corr}
 \frac{1}{\hat{A}^{\mathrm{P}}_{\mv{\lambda},n }( \mv{w} )} & =  \frac{1}{\hat{A}^{\mathrm{P}}_{n }( \mv{w} )} - \sum_{j=1}^p \lambda_{j}(\mv{w})
 \, \left( \frac{1}{\hat{A}^{\mathrm{P}}_{n}( \mv{e}_{j} )} - 1 \right), \\
\label{eq:cfg:corr}
  \log \hat{A}^{\mathrm{CFG}}_{\mv{\lambda} , n }(\mv{w}) & =  \log \hat{A}^{\mathrm{CFG}}_{n}(\mv{w}) - \sum_{j=1}^p \lambda_j(\mv{w}) \, \log \hat{A}^{\mathrm{CFG}}_{ n}(\mv{e}_{j}).
\end{align}
In case of known margins, variance-minimizing weight functions $\lambda_j$ can be determined adaptively by ordinary least squares \citep{S07Ahs, GdS11}. However, if the marginal distributions are unknown, these endpoint corrections are asymptotically irrelevant \citep{GS09}, since
\begin{align*}
  \frac{1}{\APn ( \mv{e}_{j} )} - 1 
  &= \frac{1}{n} \sum_{i=1}^{n} \log \left( \frac{n+1}{i} \right) - 1 = O\left( \frac{\log n}{n} \right) , \\
  \log \ACFGn( \mv{e}_j )
  &=  - \frac{1}{n}  \sum_{i=1}^{n} \log \left( \log \left(  \frac{n+1}{i} \right)  \right) + \int_{0}^{1} \log \left( \log \left( \frac{1}{x}  \right) \right) \diff x \\
  &= O \left( \frac{(\log n)^2}{n}  \right) .
\end{align*}
as $n \to \infty$. Nevertheless, in finite samples, the simple choice $\lambda_j(\mv{w}) = w_j$ can make quite a difference. Similarly, for unknown margins, the multivariate extension of the estimator in \citet{HT00} simplifies to $\hat{A}^{\mathrm{HT}}_n(\mv{w}) = \APn(\mv{w}) / \APn(\mv{e}_j) = \APn(\mv{w}) \{ 1 + O(n^{-1} \log n) \}$.

The next lemma establishes a functional relationship between $\APn$ and $\ACFGn$ on the one hand and the empirical copula $\hat{C}_n$ on the other hand. Recall the empirical copula process $\ecopn$ in \eqref{e:empcopproc}.

\begin{lemma}
\label{Acop:dec}
For $\mv{w} \in \simplexp$, we have
\begin{align}
 n^{1/2} \left( \frac{1}{\APn(\mv{w})} - \frac{1}{A( \mv{w} )} \right) &=   \int_{0}^{1}  \ecopn (u^{w_{1}}, \ldots , u^{w_p})  \; \frac{\diff u}{u} \label{AP:dec}, \\
n^{1/2} \bigl( \log \ACFGn ( \mv{w} ) - \log A( \mv{w} ) \bigr) &= \int_{0}^{1 }  \ecopn (u^{w_{1}}, \ldots , u^{w_p})    \frac{\diff u}{ u \log u}. \label{ACFG:dec}
\end{align}
\end{lemma}

The proof is not different from the one in dimension two and can be found in \citet[Lemma~3.1]{GS09}. Equations~\eqref{AP:dec} and \eqref{ACFG:dec} are instrumental for proving the weak convergence of the processes
\[
  \aAPn = n^{1/2} ( \APn - A )    \quad   \text{and} \quad   \aACFGn = n^{1/2} ( \ACFGn - A ).
\]

\begin{theorem}
\label{thm:main}
Under Assumption \ref{Assump:1} above, $\aAPn \dto \aAP$ and $\aACFGn \dto \aACFG$ with
\begin{eqnarray*}
  \aAP( \mv{w} )   & = &- A^{2}( \mv{w} ) \int_{0}^{1} \mathbb{C}( u^{w_{1}}, \dots , u^{w_{p}} ) \;  \frac{\diff u}{u} \\
  \aACFG( \mv{w} ) & = &   A( \mv{w} ) \int_{0}^{1} \mathbb{C}( u^{w_{1}}, \dots , u^{w_{p}} ) \frac{\diff u}{u \; \log u} ,
\end{eqnarray*}
as $n \to \infty$ in the space $\cont( \simplexp )$ equipped with the topology of uniform convergence.
\end{theorem}

The main idea of the proof consists in substituting $\ecopn$ in \eqref{AP:dec} and \eqref{ACFG:dec} by Stute's expansion and to conclude using a refined version of the continuous mapping theorem. As the proof follows the same lines as the one  in \citet{GS09},  we will just point out the main adjustments.

\section{Projection estimator}
\label{s:proj}

The estimators of the Pickands dependence functions considered so far are in general not valid Pickands dependence functions themselves. In this section, we adapt the methodology in \citet{FGS08} to project a pilot estimate $\An$ onto the set $\Aclass$ of Pickands dependence functions of $p$-variate extreme-value copulas. To this end, we view $\Aclass$ as a closed and convex subset of the real Hilbert space $L^2(\simplexp)$ with $\simplexp$ equipped with $(p-1)$-dimensional Lebesgue measure when viewed as a subset of $\RR^{p-1}$. The inner product and the norm on $L^2(\simplexp)$ are denoted by $\inpr{f}{g} = \int fg$ and $\norm{f}_2 = (\inpr{f}{f})^{1/2}$ respectively. 

The orthogonal projection of an initial estimator $\An$ for $A$, for example the Pickands or the CFG estimator, onto $\Aclass$ is then defined as
\[ 
   \Aproj = \Pi( \An | \Aclass ) = \argmin_{  A \in \Aclass }   \| A - \An \|_{2}.
\]
Projections being contractions, it follows that $\norm{ \Aproj - A }_2 \le \norm{ \An - A }_2$ for all $A \in \Aclass$: the $L^2$-risk of the projection estimator is bounded by the one of the initial estimator.




For practical computations, we are obliged to refer to finite-dimensional subclasses $\Aclassm \subset \Aclass$, yielding the approximate projection estimator 
\begin{equation}
\label{eq:Aprojm}
   \Aprojm = \Pi ( \An | \Aclassm )  = \argmin_{  A \in \Aclassm }   \| A - \An \|_{2}.
\end{equation}
For each positive integer $m$, the class $\Aclassm$ will be defined as the set of Pickands dependence functions characterized by discrete spectral measures $H$ with fixed and finite support depending on $m$. 

Specifically, let $\Vc_{p,m}$ be the (finite) set of points $\mv{v} = (v_1, \ldots, v_p) \in \Delta_{p-1}$ such that $k_j = m v_j$ is integer for every $j \in \{1, \ldots, p\}$, so that in fact $\mv{v} = (k_1/m, \ldots, k_p/m)$ where $k_j \in \{0, \ldots, m\}$ and $k_1 + \cdots + k_p = m$. The cardinality of $\Vc_{p,m}$ is of the order $O(m^{p-1})$ as $m \to \infty$. Let $\Hc_p$ be the set of spectral measures on $\Delta_{p-1}$ and let $\Hc_{p,m}$ be the set of (discrete) spectral measures $H \in \Hc_p$ supported on $\Vc_{p,m}$, that is, $H = \sum_{ \mv{v} \in \Vc_{p,m}} h_{\mv{v}} \, \delta_{\mv{v}}$, with $\delta_{\mv{v}}$ the Dirac measure at $\mv{v}$ and where $h_{\mv{v}} = H(\{ \mv{v} \})$ is the spectral mass of the atom $\mv{v}$. The vector $\mv{h} = (h_{\mv{v}})_{\mv{v} \in \Vc_{p,m}}$ satisfies the constraints
\begin{equation}
\label{constraints}
  \left\{
    \begin{array}{l@{\quad}l}
      h_{\mv{v}}  \ge  0, & \forall \mv{v} \in \Vc_{p,m} , \\[1ex]
      \sum_{\mv{v} \in \Vc_{p,m}}  h_{\mv{v}} \, v_j = 1, & \forall j \in \{ 1, \dots , p \},
    \end{array}
  \right.
\end{equation}
the second constraint stemming from~\eqref{H:constr}.

The Pickands dependence function $A$ of a spectral measure $H$ in $\Hc_{p,m}$ can be written as
\begin{equation}
\label{A:m}
  A_{\mv{h}}(\mv{w}) = \sum_{\mv{v} \in \Vc_{p,m}} h_{\mv{v}} \, \max \{w_1 v_1, \ldots, w_p v_p\}, \qquad \mv{w} \in \simplexp.
\end{equation}
Being a linear combination of piecewise linear functions, the function $A$ in \eqref{A:m} is itself piecewise linear. All Pickands dependence function of the form \eqref{A:m} will be collected in the class $\Aclassm$. The next result can be seen as the equivalent of Lemma~2 in \citet{FGS08} stating the denseness of the piecewise linear Pickands dependence functions.

\begin{lemma}
For every $H \in \Hc_p$ and every positive integer $m$, there exists $H_m \in \Hc_{p,m}$ such that the Pickands dependence functions $A$ and $A_m$ of $H$ and $H_m$ respectively satisfy
\begin{equation}
\label{eq:approx}
  \sup_{\mv{w} \in \simplexp} | A_m(\mv{w}) - A(\mv{w}) | \le \frac{p^2}{m}.
\end{equation}
\label{lemma:dens}
\end{lemma}

The bound in \eqref{eq:approx} implies that $\sup_{A \in \Aclass} \inf_{\tilde{A} \in \Aclassm} \norm{ \tilde{A} - A }_2 = O(m^{-1})$ as $m \to \infty$. This rate is perhaps not sharp, for in case $p = 2$, Lemma~2 in \citet{FGS08} states the rate $O(m^{-3/2})$. It remains an open problem whether the latter rate can also be achieved in general dimension $p$.


In practice, the task is to compute the vector $\hat{\mv{h}}$ such that the function
\[
\Aprojm (\mv{w}) = A_{\hat{\mv{h}}}(\mv{w}) =  \sum_{\mv{v} \in \Vc_{p,m}}  \hat{h}_{\mv{v}} \max \{ w_{1} v_{1} , \dots , w_{p} v_{p} \}, \quad \mv{w}  \in \simplexp,
\]
solves \eqref{eq:Aprojm}. The vector $\hat{\mv{h}}$ is given as the solution to the least-squares problem
\begin{equation}
 \label{h:optim}
  \hat{\mv{h}} = \argmin_{\mv{h}} \norm{ A_{\mv{h}} - \An }_2^2 = \argmin_{\mv{h}} \bigl( \norm{A_{\mv{h}}}_2^2 - 2 \inpr{A_{\mv{h}}}{\An} \bigr),
\end{equation}
with $\mv{h}$ subject to the constraints \eqref{constraints}. The optimisation problem in \eqref{h:optim} is a quadratic program with linear constraints, which in matrix notation reads
\begin{equation}
\label{eq:quadprog}
  \mv{\hat{h}}  = \argmin_{\mv{h}} \biggl( \frac{1}{2} \mv{h}^{\top} \mv{D}  \mv{h} -  \mv{d}^{\top} \mv{h} \biggr), 
  \quad \text{subject to} \quad
\begin{cases}
\mv{C}\mv{h} = \mv{c}, \\
\mv{h} \ge \mv{0}.
\end{cases}
\end{equation}
The matrix $\mv{D}$ and the vector $\mv{d}$ regroup all the scalar products of the form 
\[
  \int_{\simplexp } \max ( \mv{w} \, \mv{v}) \, \max ( \mv{w} \, \mv{v}') \, \diff \mv{w}
  \qquad \text{and} \qquad 
  \int_{\simplexp}   \max ( \mv{w} \, \mv{v}) \, \An ( \mv{w} ) \, \diff \mv{w}
\]
respectively, for $\mv{v}, \mv{v}' \in \Vc_{p,m}$. The $p$ equality constraints $\sum_{\mv{v} \in \Vc_{p,m}}  h_{\mv{v}} \, v_j = 1$ are encoded by means of the matrix $\mv{C}$ and the vector $\mv{c}$. 

For implementation, we used the R-package \textsf{quadprog} \citep{W10} for solving quadratic programs under linear constraints. Although there exist multiple packages for numerical multivariate integration, we preferred to compute all the integrals appearing in $\mv{D}$ and $\mv{d}$ using Riemann sums on the same fine grid. By so doing we reduce the risk of numerical problems as we impose $\mv{D}$ to be positive definite.


The derivation of the asymptotics of the projection estimator follows the same lines as in \citet{FGS08}. Assume that $\eps_n^{-1} (\An - A) \dto \zeta$ in $L^2(\Delta_{p-1})$ where $\zeta$ is a random process in $L^2(\Delta_{p-1})$ and $0 < \eps_n \to 0$; for the Pickands and CFG estimators, we have $\eps_n = n^{-1/2}$ and we have weak convergence with respect to the uniform topology, which implies convergence with respect to the $L^2$-topology.

By Lemma~1 in \citet{FGS08}, we have
\[
  \norm{ \Aprojm - \Aproj }_2 \le [ \delta_{m} \{  2 \norm{ \An - \Aproj }_2  + \delta_{m} \} ]^{1/2},
\]
with $\delta_{m} =  \norm{ \Aproj - \Pi ( \Aproj \vert \Aclassm ) }_2 $. From Lemma~\ref{lemma:dens} above, we have $\delta_m = O(1/m)$ as $m \to \infty$. As a consequence, if $m = m_n$ is such that $1/m_n = o(\eps_n)$ as $n \to \infty$, then $\norm{ \Aprojm - \Aproj }_2 = o_{\mathrm{P}}( \epsilon_{n}) $. From \citet[Theorem~1]{FGS08}, we deduce that
\begin{equation}
\label{eq:proj:conv}
  \eps_n^{-1} (\Aprojm - A ) = \eps_n^{-1} (\Aproj - A) + o_{\mathrm{P}}(1) \dto \Pi \bigl( \zeta | \mathrm{T}_{\Aclass}(A) \bigr)
  \qquad (n \to \infty)
\end{equation}
in the space $L^{2}( \simplexp )$, where $\mathrm{T}_{\Aclass}(A)$ is the tangent cone of $\Aclass$ at $A$, defined as the $L^2$-closure of $\{ \lambda (  \tilde{A} - A ): \lambda \ge 0 , \tilde{A} \in \Aclass \}$. 

Interestingly, equation~\eqref{eq:proj:conv} implies that the choice of $m$ is not to be seen as a bias-variance trade-off problem but rather as a discretization problem. As soon as $m = m_n$ converges to infinity faster than $1/\eps_n$, the finite-dimensional projection estimator $\Aprojm$ has the same limit behaviour as the `ideal' projection estimator $\Aproj$. In practice, we will choose $m$ sufficiently large so that any further increase of $m$ does not make any significant difference, of course subject to constraints on computing time and numerical stability.

Finally, note that the convergence in \eqref{eq:proj:conv} is with respect to the $L^2$-topology only, even if originally the weak convergence of $\eps_n^{-1} (\An - A)$ took place in the stronger $\ell^\infty$-topology. The asymptotic distribution of the projection estimator under the $\ell^\infty$-topology remains an open problem.

\section{Simulation study}
\label{s:simul}

A simulation experiment was conducted to compare the finite-sample performance of the following four estimators: 
\begin{compactitem}
\item[PD --] the endpoint-corrected Pickands estimator in \eqref{eq:pick:corr} with $\lambda_j(\mv{w}) = w_j$, in the spirit of \citet{Deheuvels91};
\item[PD-pr --] the projection estimator in \eqref{eq:Aprojm} with the previous end-point corrected Pickands estimator as initial estimator;
\item[CFG --] the endpoint-corrected CFG estimator in \eqref{eq:cfg:corr} with $\lambda_j(\mv{w}) = w_j$;
\item[CFG-pr --] the projection estimator in \eqref{eq:Aprojm} with the previous end-point corrected CFG estimator as initial estimator.
\end{compactitem}

The set-up of the experiment was as follows. Following \citet{ZWP08} and \citet{GdS11}, random samples were generated from a trivariate extreme-value distribution with asymmetric logistic dependence function $A$ \citep{tawn1990}:
\begin{multline}
\label{E:logistic}
  A(\mv{w}) = 
  (\theta^{1/\alpha} w_{1}^{1/\alpha} + \phi^{1/\alpha} w_{2}^{1/\alpha} )^{\alpha} + (\theta^{1/\alpha} w_{2}^{1/\alpha} + \phi^{1/\alpha} w_{3}^{1/\alpha} )^{\alpha} \\
  + (\theta^{1/\alpha} w_{3}^{1/\alpha} + \phi^{1/\alpha} w_{1}^{1/\alpha} )^{\alpha} + \psi ( w_{1}^{1/\alpha} + w_{2}^{1/\alpha} + w_{3}^{1/\alpha}   )^{\alpha} 
  + 1 - \theta - \phi -\psi,
\end{multline}
for $\mv{w} \in \simplextwo$, with parameter vector $( \alpha, \theta, \phi, \psi) \in (0,1] \times [0, 1]^3$. For this model, Assumption~\ref{Assump:2} can be verified by direct calculation. The dependence parameter $\alpha$ ranged from $0.3$ (high dependence) to $1$ (independence, $A \equiv 1$) and the vector $(\phi, \psi, \theta)$ was set equal to either $(0, 1, 0)$ (symmetric logistic copula or Gumbel copula) and $(0.3, 0, 0.6)$ (an asymmetric logistic copula). For each distribution, $1000$ samples were generated of size $n \in \{ 50, 100, 200 \}$. Simulations were performed using the R-package \textsf{evd} \citep{evd}, which implements the algorithms presented in \citet{Stephenson03}. The discretization parameter $m$ was set to $20$, at which value the grid $\Vc_{3, 20}$ contains $231$ points.

Monte-Carlo approximations for the mean integrated squared error (MISE) $\E [ \int (\hat{A} - A)^2 ]$ for the four estimators considered above are reported in the tables below. The three main findings are the following:
\begin{compactenum}[1-]
\item The projection step yields a gain in efficiency, especially in case of weak dependence.
\item Without projection step, the CFG estimator outperforms the PD estimator.
\item After the projection step, the PD-pr estimator is more efficient than the CFG-pr estimator in case of independence and weak dependence ($\alpha \ge 0.9$), but less efficient otherwise ($\alpha \le 0.7$).
\end{compactenum}
Further, as the dependence increases, all estimators tend to perform better. In accordance with asymptotic theory, the MISE is roughly proportional to $1/n$.

\begin{table}[H]
\begin{center}\small
\begin{tabular}{l@{\quad}l@{\quad}ccccc}
\toprule
	    & &  $\alpha=0.3$ & $\alpha=0.5$  & $\alpha=0.7$  & $\alpha=0.9$ & $\alpha = 1$ \\
\midrule
$n=50$  & PD &	$1.40 \cdot 10^{-4}$	& $5.44 \cdot 10^{-4}$	& $1.36 \cdot 10^{-3}$	& $2.68 \cdot 10^{-3}$ & $3.44 \cdot 10^{-3}$ \\
	    & PD-pr & $1.37 \cdot 10^{-4}$	& $5.14 \cdot 10^{-4}$	& $1.21 \cdot 10^{-3}$	& $2.08 \cdot 10^{-3}$ & $2.44 \cdot 10^{-3}$ \\[1ex]
  & CFG         &	$9.77 \cdot 10^{-5}$	& $4.27 \cdot 10^{-4}$	& $1.26 \cdot 10^{-3}$	& $2.54 \cdot 10^{-3}$ & $3.48 \cdot 10^{-3}$ \\
	    & CFG-pr  & $9.69 \cdot 10^{-5}$	& $4.22 \cdot 10^{-4}$	& $1.22 \cdot 10^{-3}$	& $2.43 \cdot 10^{-3}$ & $3.31 \cdot 10^{-3}$ \\[1em]
$n=100$ & PD &	$7.08 \cdot 10^{-5}$	& $2.84 \cdot 10^{-4}$	& $7.06 \cdot 10^{-4}$	& $1.34 \cdot 10^{-3}$ & $1.69 \cdot 10^{-3}$ \\
	    & PD-pr & $6.99 \cdot 10^{-5}$	& $2.74 \cdot 10^{-4}$	& $6.53 \cdot 10^{-4}$	& $1.03 \cdot 10^{-3}$ & $1.08 \cdot 10^{-3}$ \\[1ex]
 & CFG         &	$5.03 \cdot 10^{-5}$	& $2.39 \cdot 10^{-4}$	& $6.56 \cdot 10^{-4}$	& $1.23 \cdot 10^{-3}$ & $1.48 \cdot 10^{-3}$ \\
	    & CFG-pr  & $5.01 \cdot 10^{-5}$	& $2.37 \cdot 10^{-4}$	& $6.47 \cdot 10^{-4}$	& $1.18 \cdot 10^{-3}$ & $1.36 \cdot 10^{-3}$ \\[1em]
$n=200$	& PD & $3.31 \cdot 10^{-5}$	& $1.43 \cdot 10^{-4}$	& $3.92 \cdot 10^{-4}$	& $7.02 \cdot 10^{-4}$ & $8.71 \cdot 10^{-4}$ \\
	    & PD-pr & $3.29 \cdot 10^{-5}$	& $1.40 \cdot 10^{-4}$	& $3.73 \cdot 10^{-4}$	& $5.72 \cdot 10^{-4}$ & $5.14 \cdot 10^{-4}$ \\[1ex]
 & CFG         & $2.45 \cdot 10^{-5}$  & $1.23 \cdot 10^{-4}$	& $3.39 \cdot 10^{-4}$	& $6.40 \cdot 10^{-4}$ & $6.56 \cdot 10^{-4}$ \\
	    & CFG-pr  & $2.45 \cdot 10^{-5}$	& $1.23 \cdot 10^{-4}$	& $3.36 \cdot 10^{-4}$	& $6.19 \cdot 10^{-4}$ & $5.78 \cdot 10^{-4}$ \\
\bottomrule
\end{tabular}
\end{center}
\caption{Symmetric logistic dependence function, $(\phi, \psi, \theta) = (0, 1, 0)$: Monte-Carlo approximation of the MISE of four estimators of $A$ based on $1000$ random samples}
\label{sym}
\end{table}

\begin{table}[H]
\begin{center}\small
\begin{tabular}{l@{\qquad}l@{\qquad}cccc}
\toprule
    &    &  $\alpha=0.3$ & $\alpha=0.5$  & $\alpha=0.7$  & $\alpha=0.9$ \\
\midrule
$n=50$  & PD	& $ 1.42 \cdot 10^{-3}$   & $1.72 \cdot 10 ^{-3}$ & $2.20 \cdot 10^{-3}$  & $2.88 \cdot 10^{-3}$ \\
	    & PD-pr & $1.22 \cdot 10^{-3}$	& $1.45 \cdot 10^{-3}$	  & $1.74 \cdot 10{-3}$   & $2.01 \cdot 10^{-3}$ \\[1ex]
  &  CFG         &	$1.15 \cdot 10^{-3}$ & 	$1.41 \cdot 10^{-3}$	& $1.84 \cdot 10^{-3}$  &	$2.77 \cdot 10^{-3}$\\
	    &  CFG-pr  &  $1.10 \cdot 10^{-3}$ &	 $1.35 \cdot 10^{-3}$ & 	$1.75 \cdot 10^{-3}$	& $2.60 \cdot 10^{-3}$ \\[1em]
$n=100$	& PD & $7.67 \cdot 10^{-4}$	& $8.76 \cdot 10^{-4}$	  & $1.10 \cdot 10^{-3}$	& $1.51 \cdot 10^{-3}$ \\
	    & PD-pr & $6.77 \cdot 10^{-4}$	& $7.70 \cdot 10^{-4}$	  & $9.00 \cdot 10^{-4}$	& $1.08 \cdot 10^{-3}$ \\[1ex]
 &  CFG	         &  $5.90 \cdot 10^{-4}$	& $7.06 \cdot 10^{-4}$ & 	$9.05 \cdot 10^{-4}$	& $1.20 \cdot 10^{-3}$ \\
	    &  CFG-pr  & $5.70 \cdot 10^{-4}$ & 	$6.85 \cdot 10^{-4}$ &	$8.68 \cdot 10^{-4}$ &	$1.10 \cdot 10^{-3}$ \\[1em]
$n=200$	& PD & $3.92 \cdot 10^{-4}$	& $4.72 \cdot 10^{-4}$	  & $5.84 \cdot 10^{-4}$	& $7.60 \cdot 10^{-4}$ \\
      & PD-pr  & $3.52 \cdot 10^{-4}$	& $4.30 \cdot 10^{-4}$	  & $5.08 \cdot 10^{-4}$	& $5.19 \cdot 10^{-4}$ \\[1ex]
&  CFG          & $3.01 \cdot 10^{-4}$ & 	$3.31 \cdot 10^{-4}$ &	$4.43 \cdot 10^{-4}$ &	$5.81 \cdot 10^{-4}$\\
	    &  CFG-pr  & $2.92 \cdot 10^{-4}$ & 	$3.22 \cdot 10^{-4}$ &  $4.29 \cdot 10^{-4}$	& $5.36 \cdot 10^{-4}$\\
\bottomrule
\end{tabular}
\end{center}
\caption{Asymmetric logistic dependence function, $(\phi, \psi, \theta) = (0.3, 1, 0.6)$: Monte-Carlo approximation of the MISE of four estimators of $A$ based on $1000$ random samples}
\label{asym}
\end{table}

\small
\bibliographystyle{chicago}
\bibliography{MVEC-proj}

\appendix
\section{Proofs}
\label{s:proofs}

\subsection{Proof of Proposition~\ref{thm:empcopproc}}
If $\mv{u} \in  (0, 1]^p$, then $-\log \mv{u} \in W_{p,j}$ and
\begin{equation}
\label{eq:Cdotj}
  \dot{C}_j(\mv{u}) = \frac{C(\mv{u})}{u_j} \, \dot{\ell}_j(-\log \mv{u}).
\end{equation}
The assumptions on $\ell$ imply continuity of $\dot{C}_{j}$ on the set $(0, 1]^p$. If $\mv{u} \in [0, 1]^p$ with $u_j > 0$ and $u_i = 0$ for some $i \in \{1, \ldots, p\} \setminus \{j\}$, then $\dot{C}_j(u) = 0$ and continuity of $\dot{C}_{j}$ at such $\mv{u}$ follows from the fact that $0 \le \dot{\ell}_j \le 1$ and $0 \le C(\mv{v}) \le \min(\mv{v})$.

If (L2) holds, then also
\[
  \sup_{\mv{x} \in W_{p,i} \cap W_{p,j}} \max(x_i, x_j) \, |\ddot{\ell}_{i,j}( \mv{x} )| < \infty,
\]
that is, without the condition $x_1 + \cdots + x_p = 1$.
This result is based on the fact that the function $\ell$ is homogeneous of order one: $\ell(s \mv{x} ) = s \, \ell(\mv{x})$ for all $s \in [0, \infty)$ and $\mv{x} \in [0, \infty)^p$. Hence if $\dot{\ell}_j$ exists on $W_{p,j}$, then for all $s \in (0, \infty)$ and $\mv{x} \in W_{p,j}$ we have
\[
  s \, \dot{\ell}_j(s \mv{x} ) = \frac{\partial}{\partial x_j} \ell(s \mv{x} ) = \frac{\partial}{\partial x_j} s \, \ell(\mv{x}) = s \, \dot{\ell}_j(\mv{x})
\]
and thus
\[
  \dot{\ell}_j(s\mv{x}) = \dot{\ell}_j(\mv{x}).
\]
Taking partial derivatives again, we find for all $s \in (0, \infty)$ and $\mv{x} \in W_{p,i} \cap W_{p,j}$ that
\[
  s \, \ddot{\ell}_{ij}(s\mv{x}) = \frac{\partial}{\partial x_i} \dot{\ell}_j(s\mv{x}) = \frac{\partial}{\partial x_i} \dot{\ell}_j(\mv{x}) = \ddot{\ell}_{ij}(\mv{x})
\]
and thus
\[
  \ddot{\ell}_{ij}(s\mv{x}) = s^{-1} \, \ddot{\ell}_{ij}(\mv{x}).
\]
It follows that
\[
  \max(sx_i, sx_j) \, \ddot{\ell}_{ij}(s\mv{x}) = \max(x_i, x_j) \, \ddot{\ell}_{ij}(\mv{x}),
\]
that is, the map $\mv{x} \mapsto \max(x_i, x_j) \, \ddot{\ell}_{ij}(\mv{x})$ is constant on rays through the origin. 

Next, we show the equivalence of (L2) and (C2). Fix $i, j \in \{1, \ldots, p\}$, not necessarily distinct and  let $\mv{u} \in V_{p,i} \cap V_{p,j}$. On the one hand, if $\mv{u} \in V_{p,i} \cap V_{p,j} \cap (0, 1]^p$ (meaning that every component of $\mv{u}$ is different from 0), then $-\log \mv{u} \in W_{p,i} \cap W_{p,j}$ and
\[
  \ddot{C}_{ij}(\mv{u}) = 
  \begin{cases}
    \displaystyle \frac{C(\mv{u})}{u_i u_j} \, \bigl( \dot{\ell}_i \dot{\ell}_j - \ddot{\ell}_{ij} \bigr) & \text{if $i \ne j$}, \\[1em]
    \displaystyle \frac{C(\mv{u})}{u_j^2} \, \bigl( \dot{\ell}_j^2 - \dot{\ell}_j - \ddot{\ell}_{jj} \bigr) & \text{if $i = j$},
  \end{cases}
\]
with the convention that the partial derivatives of $\ell$ are evaluated in $-\log \mv{u}$. On the other hand, if $\mv{u} \in (V_{p,i} \cap V_{p,j}) \setminus (0, 1]^p$ (i.e.\ at least one coordinate of $\mv{u}$ vanishes), then $\ddot{C}_{ij}(\mv{u}) = 0$.

We have to verify two things: first, the continuity of $\ddot{C}_{ij}$ at points in the set $(V_{p,i} \cap V_{p,j}) \setminus (0, 1]^p$; secondly, the finiteness of the supremum in (C2).

First, let $\mv{u} \in V_{p,i} \cap V_{p,j} \cap (0, 1]^p$. Let $K$ be a positive constant not smaller than the supremum in (L2). By assumption (L2) and  the fact that $0 \le \dot{\ell}_j \le 1$, we have
\[
  |\ddot{C}_{ij}(\mv{u})| \le \frac{\min(\mv{u})}{u_i u_j} \, \biggl( 1 + \frac{K}{\max(-\log u_i, -\log u_j)} \biggr).
\]
Continuity of $\ddot{C}_{ij}$ at points in the set $V_{p,i} \cap V_{p,j} \setminus (0, 1]^p$ follows. 

Secondly, as $\min(\mv{u}) / (u_i u_j) \le \min(1/u_i, 1/u_j)$ and $- \log x \ge 1-x$ for all positive $x$, 
\begin{align*}
  |\ddot{C}_{ij}(\mv{u})| 
  &\le \min \biggl(\frac{1}{u_i}, \frac{1}{u_j}\biggr) \, \biggl\{ 1 + K \, \min \biggl( \frac{1}{1-u_i}, \frac{1}{1-u_j} \biggr) \biggr\} \\
  &\le (1+K) \, \min \biggl(\frac{1}{u_i}, \frac{1}{u_j}\biggr) \, \min \biggl( \frac{1}{1-u_i}, \frac{1}{1-u_j} \biggr) \\
  &\le (1+K) \, \min \biggl( \frac{1}{u_i(1-u_i)}, \, \frac{1}{u_j(1-u_j)} \biggr),
\end{align*}
which is equivalent to condition (C2).

\subsection{Proof of Theorem \ref{thm:main}}

The proof of theorem \ref{thm:main} will require the following preliminary result on weighted empirical copula processes. Recall the process $\alpha_n$ in \eqref{e:empproc}. Define $q_{\theta }(t) = t^{\theta} (1-t)^{\theta }$ for $t \in (0,1)$ and  a fixed value $\theta \in (0,1/2)$. Write $\mathbb{E} = (0,1]^{p} \setminus \{ (1, \dots , 1) \}$. Define the process $\mathbb{G}_{n,\theta}$ on $[0, 1]^p$ by
\begin{equation}
\mathbb{G}_{n, \theta}(\mv{u}) = 
  \begin{cases}
    \frac{\alpha_{n}( \mv{u} )}{q_{\theta}(\min(\mv{u})) }
    & \text{if $\mv{u} \in \mathbb{E}$}, \\
    0 & \text{if $\mv{u} \in [0, 1]^p \setminus \mathbb{E}$.}
  \end{cases}
\label{Gng}
\end{equation}
Similarly, define the process $\mathbb{G}_\theta$ on $[0, 1]^p$ by replacing $\alpha_n$ in \eqref{Gng} by its weak limit $\alpha$, see \eqref{e:alpha}. The following result generalizes Theorem~G.1 in \citet{GS09}.

\begin{lemma}
For every $\theta \in (0, 1/2)$, the trajectories of $\mathbb{G}_\theta$ are continuous almost surely and $\mathbb{G}_{n,\theta} \dto \mathbb{G}_\theta$ in $\ell^\infty([0,1]^p)$.
\label{lemmaGn}
\end{lemma}

\begin{proof}[Lemma~\ref{lemmaGn}]
The proof is entirely analogue as the one of Theorem~G.1 in \citet{GS09}. For completeness, we sketch the main lines.

Fix $\mv{u} \in \mathbb{E}$ and define the mapping $ f_{\mv{u}}: \mathbb{E} \to \mathbb{R} $ by
\begin{equation*}
 f_{\mv{u}}(\mv{s}) = \frac{\mv{1}_{(\mv{0}, \mv{u}]} (\mv{s})  - C(\mv{u})}{q_{ \theta }( \min(\mv{u}))}, \qquad \mv{s} \in \mathbb{E},
\end{equation*}
and consider the class
\begin{equation*}
\mathcal{F} = \{ f_{\mv{u}}:  \mv{u} \in \mathbb{E}   \} \cup \{ 0 \},
\end{equation*}
where $0$ of course stands for the zero function. The space $\mathcal{F}$ will be endowed with the metric 
\begin{equation}
\label{eq:metric}
  \rho^{2}(f,g) = \proba(f-g)^{2} \qquad f, g \in \mathcal{F}.
\end{equation}
Here, we adopt the notations of \citet{VW96}: $\proba$ denotes the probability distribution on $\mathbb{E}$ corresponding to $C$ and $\proban$ denotes the  empirical measure of the sample $(U_{i1} , \dots , U_{ip})$ for $i \in \{1, \ldots, n\}$, that is
\begin{align*}
\proba f  & =  \int f \, \diff C, & \proban f & =  \frac{1}{n} \sum_{i=1}^{n} f( U_{i,1}, \dots , U_{i,p} ).
\end{align*}
Moreover, put $\mathbb{G}_{n} = n^{1/2} ( \proban - \proba )$, viewed as a random function on $\mathcal{F}$. 

We will show that the collection $\mathcal{F}$ is a $\proba$-Donsker class, i.e.\ there exists a $\proba$-Brownian bridge  $\mathbb{G}$ such that
\[
\mathbb{G}_{n} \dto \mathbb{G} \quad \text{in} \quad \ell^{ \infty } (\mathcal{F}) 
\; \text{as} \; n \to \infty.
\]
It is sufficient to verify the conditions of Theorem~2.6.14 of \citet{VW96}. The function $F$ on $\mathbb{E}$ defined by
\begin{equation*}
F( s_{1}, \dots , s_{p }) = p \max \bigl\{ s_{1}^{-\theta}, \ldots, s_{p}^{-\theta}, (1-s_{1})^{-\theta}, \ldots, (1- s_{p})^{-\theta} \bigr\}
\end{equation*}
is a suitable envelope function for $\mathcal{F}$. The fact that $\mathcal{F}$ is a VC-major class and is pointwise separable follows from the same arguments as in \citet{GS09}.

For the moment $\mathbb{G}_n$ is defined on $\mathcal{F}$ with the  metric $\rho$ in \eqref{eq:metric}.  Consider the map $\phi : [0, 1]^p \to \mathcal{F}$ defined by $\phi(\mv{u}) =  f_{\mv{u}}$ if $\mv{u} \in \mathbb{E}$ and $\phi(\mv{u}) = 0$ if $\mv{u} \in [0, 1]^p \setminus \mathbb{E}$. Then $\mathbb{G}_{n,\theta} = \mathbb{G}_n \circ \phi$ and $\mathbb{G}_\theta = \mathbb{G} \circ \phi$. The map $\ell^\infty(\mathcal{F}) \to \ell^\infty([0, 1]^p) : z \mapsto z \circ \phi$ being continuous, the continuous mapping theorem permits to conclude that $\mathbb{G}_{n,\theta} \dto \mathbb{G}_{\theta} $ in $\ell^{\infty}([0,1]^p)$. Since the trajectories of $\mathbb{G}$ are $\rho$-continuous almost surely and since $\phi$ is continuous, it follows that the sample paths of $\mathbb{G}_\theta$ are continuous almost surely as well. This concludes the proof of Lemma~\ref{lemmaGn}.
\end{proof}

We now proceed with the proof of Theorem~\ref{thm:main}. Define
\begin{align*}
  B^{\mathrm{P}}_{n}( \mv{w} ) &= n^{1/2} \left( \frac{1}{\APn(\mv{w})} - \frac{1}{A( \mv{w} )} \right), \\
  B^{\mathrm{CFG}}_{n} ( \mv{w} ) &= n^{1/2} \bigl( \log \ACFGn ( \mv{w} ) - \log A( \mv{w} ) \bigr),
\end{align*}
for $\mv{w} \in \simplexp$. Applying the change of variables $u = e^{-s}$ in Lemma~\ref{Acop:dec}, we find that the processes $B^{\mathrm{P}}_{n}$ and $B^{\mathrm{CFG}}_{n}$ can be written as
\begin{equation}
\label{eq:Bn}
  B_n(\mv{w}) = \int_0^\infty \mathbb{C}_{n}(e^{-w_{1}s} , \dots , e^{-w_{p}s}) \, h(s) \, \diff s,
\end{equation}
in terms of a function $h$ on $(0, \infty)$ which is $h^{\mathrm{P}}(s) = 1$ for the Pickands estimator and $h^{\mathrm{CFG}}(s) = 1/s$ for the CFG estimator. In what follows, the function $h$ denotes either $h^{\mathrm{P}}$ or $h^{\mathrm{CFG}}$.

Put $l_n = 1/(n+1)$ and $k_n = p \log(n+1)$ and split the integral on the right-hand side of \eqref{eq:Bn} into three parts:
\begin{equation}
\label{eq:Bn:decomp}
  B_n(\mv{w}) = \int_0^{l_n} + \int_{l_n}^{k_n} + \int_{k_n}^\infty = I_{1,n}(\mv{w}) + I_{2,n}(\mv{w}) + I_{3,n}(\mv{w}).
\end{equation}
We will first prove that with probability one, the first and the third term on the right-hand side converge to zero uniformly in $\mv{w}$. 
\begin{itemize}
\item If $s \in [0, l_n]$, then $e^{-s} \ge 1 - l_n$ and thus $C_n(e^{-w_{1}s} , \dots , e^{-w_{p}s}) = 1$, which implies
\begin{align*}
  0 \le I_{1,n}(\mv{w}) 
  &= \int_0^{l_n} n^{1/2} \bigl( 1 - e^{- s \, A(\mv{w})} \bigr) \, h(s) \, \diff s \\
  &\le n^{1/2} \int_0^{l_n} s \, h(s) \, ds \le n^{1/2} l_n \le n^{-1/2}.
\end{align*}
\item If $s \in (k_n, \infty)$, then $w_j \ge 1/p$ and thus $e^{-w_j s} < 1/(n+1)$ for at least one $j \in \{1, \ldots, p\}$, so that $C_n(e^{-w_{1}s} , \dots , e^{-w_{p}s}) = 0$, which implies
\begin{align*}
  |I_{3,n}(\mv{w})|
  &\le \int_{k_n}^\infty n^{1/2} e^{-s \, A(\mv{w})} \, h(s) \, \diff s \\
  &\le \frac{n^{1/2}}{A(\mv{w})} e^{-k_n \, A(\mv{w})} \le p \, n^{-1/2},
\end{align*}
where we used the fact that $A(\mv{w}) \ge \max(\mv{w}) \ge 1/p$.
\end{itemize}

As a consequence, the only non-negligible term in \eqref{eq:Bn:decomp} is $I_{2,n}$. By Assumption~\ref{Assump:1} and by Proposition~4.2 in \citet{Segers11}, Stute's expansion \eqref{e:stute}--\eqref{e:stute:Rn} is valid, so that we can write
\[
  I_{n,2}(\mv{w}) = J_{0,n}(\mv{w}) - \sum_{j=1}^p J_{j,n}(\mv{w}) + J_{p+1,n}(\mv{w})
\]
where
\begin{align*}
  J_{0,n}(\mv{w}) &= \int_{l_n}^{k_n} \alpha_n(e^{-w_{1}s} , \dots , e^{-w_{p}s}) \, h(s) \, \diff s, \\
  J_{j,n}(\mv{w}) &= \int_{l_n}^{k_n} \alpha_{n,j}(e^{-w_{j}s}) \, \dot{C}_j(e^{-w_{1}s} , \dots , e^{-w_{p}s}) \, h(s) \, \diff s, \qquad j \in \{1, \ldots, p\}, \\
  J_{p+1,n}(\mv{w}) &= \int_{l_n}^{k_n} R_n(e^{-w_{1}s} , \dots , e^{-w_{p}s}) \, h(s) \, \diff s.
\end{align*}
In view of the bound \eqref{e:stute:Rn} on $R_n$, the term $J_{p+1,n}$ is negligible: as $n \to \infty$,
\[
  \sup_{\mv{w} \in \simplexp} |J_{p+1,n}(\mv{w})|
  = O \bigl(n^{-1/4} \, \log(n) \, {\textstyle\int_{l_n}^{k_n} h(s) \, \diff s} \bigr)
  \to 0, \qquad \text{almost surely}.
\]

Fix $\theta \in (0, 1/2)$ and recall the process $\mathbb{G}_{n,\theta}$ in \eqref{Gng}. We have
\begin{align*}
  J_{0,n}(\mv{w}) &= \int_{l_n}^{k_n} \mathbb{G}_{n,\theta}(e^{-sw_1}, \ldots, e^{-sw_p}) \, K_0(s, \mv{w}) \, h(s) \, \diff s, \\
  J_{j,n}(\mv{w}) &= \int_{l_n}^{k_n} \mathbb{G}_{n,\theta}(1, \ldots, 1, e^{-sw_j}, 1, \ldots, 1) \, K_j(s, \mv{w}) \, h(s) \, \diff s, \qquad j \in \{1, \ldots, p\},
\end{align*}
with
\begin{align*}
  K_0(s, \mv{w}) &= q_\theta \bigl( \min( e^{-sw_1}, \ldots, e^{-sw_p} ) \bigr), \\
  K_j(s, \mv{w}) &= q_\theta( e^{-sw_j} ) \, \dot{C}_j( e^{-sw_1}, \ldots, e^{-sw_p} ), \qquad j \in \{1, \ldots, p\}.
\end{align*}
The functions $K_0, \ldots, K_p$ satisfy the bounds
\begin{multline}
\label{eq:Kbound}
  0 \le K_j(s, \mv{w}) \le K(s) = s^\theta \, \1_{(0, 1]}(s) + e^{ - (\theta/p) s } \, \1_{(1, \infty)}(s), \\
  \qquad j \in \{0, \ldots, p\}, \qquad s \in (0, \infty), \qquad \mv{w} \in \simplexp.
\end{multline}
To prove these bounds, use equation~\eqref{eq:Cdotj}, the fact that $0 \le \dot{\ell}_j \le 1$ and $0 \le C(\mv{v}) \le \min(\mv{v})$ for $\mv{v} \in [0, 1]^p$ and the fact that $\max(\mv{w}) \ge 1/p$ for $\mv{w} \in \simplexp$. The function $K$ in \eqref{eq:Kbound} satisfies $\int_0^\infty K(s) \, h(s) \, \diff s < \infty$.

By Lemma~\ref{lemmaGn} and the extended continuous mapping theorem \citep[Theorem~1.11.1]{VW96}, we find
\[
  B_n \dto B = J_0 - \sum_{j=1}^p J_j, \qquad n \to \infty
\]
in $\ell^\infty(\simplexp)$, where
\begin{align*}
  J_0(\mv{w}) &= \int_0^\infty \mathbb{G}_\theta(e^{-sw_1}, \ldots, e^{-sw_p}) \, K_0(s, \mv{w}) \, h(s) \, \diff s, \\
  J_j(\mv{w}) &= \int_0^\infty \mathbb{G}_\theta(1, \ldots, 1, e^{-sw_j}, 1, \ldots, 1) \, K_j(s, \mv{w}) \, h(s) \, \diff s, \qquad j \in \{1, \ldots, p\}.
\end{align*}
Substituting the definitions of the process $\mathbb{G}_\theta$ and the functions $K_0, \ldots, K_p$, we obtain
\[
  B(\mv{w}) = \int_0^\infty \mathbb{C}(e^{-w_{1}s} , \dots , e^{-w_{p}s}) \, h(s) \, \diff s, \qquad \mv{w} \in \simplexp.
\]
An application of the functional delta method \citep[Theorem~3.9.4]{VW96} now yields the result.

\subsection{Proof of Lemma \ref{lemma:dens}}

The proof is constructive and consists of the following steps:
\begin{compactenum}[1.]
\item Construction of the spectral measure $H_m$.
\begin{compactenum}[(a)]
\item Discretisation of $H$ yielding a measure $G_m$ on $\Delta_{p-1}$, which is not necessarily a spectral measure.
\item Modification of $G_m$ into a genuine spectral measure $H_m$ .
\end{compactenum}
\item Proof of the inequality~\eqref{eq:approx}.
\end{compactenum} \bigskip

\noindent\textit{1. Construction of the spectral measure $H_m$.}
For $\mv{v} \in \Vc_{p,m}$, consider the set $\Delta_{p-1,\mv{v},m}$ of points $\mv{t} \in \simplexp$ such that $v_j \le t_j < v_j + 1/m$ for every $j \in \{1, \ldots, p-1\}$; recall that $t_p = 1 - t_1 - \cdots - t_{p-1}$, so that necessarily $v_p - (p-1)/m < t_p \le v_p$. The collection of sets $\{ \Delta_{p-1,\mv{v},m} : \mv{v} \in \Vc_{p,m} \}$ constitutes a partition of $\Delta_{p-1}$. Indeed, for every point $\mv{t} \in \simplexp$ there is a unique point $\mv{v} \in \Vc_{p,m}$ such that $\mv{t} \in \Delta_{p-1,\mv{v},m}$: Let $v_j$ be the integer part of $m t_{j}$ for $j \in \{1, \ldots, p-1\}$ and put $v_p = 1 - v_1 - \cdots - v_{p-1}$.

\textit{(a) Discretisation of $H$, yielding $G_m$.} Define a discrete measure $G_m$ on $\simplexp $ with support contained in $\Vc_{p,m}$ by $G_m(\{ \mv{v} \}) = H(\Delta_{p-1,\mv{v},m})$ for $\mv{v} \in \Vc_{p,m}$. In words, the mass assigned by the spectral measure $H$ on the set $\Delta_{p-1,\mv{v},m}$ is relocated to the corner point $\mv{v}$.

Since the sets $\Delta_{p-1,\mv{v},m}$ constitute a partition of $\simplexp$, the total mass of $G_m$ is still $G_m(\simplexp) = H(\simplexp ) = p$. However, $G_m$ does not need to verify the moment constraints. For $j \in \{1, \ldots, p\}$ we have
\[
  \int_{\Delta_{p-1}} t_j \, dG_m( \mv{t}) 
  = \sum_{\mv{v} \in \Vc_{p,m}} v_j \, G_m(\{ \mv{v} \}) 
  = \sum_{\mv{v} \in \Vc_{p,m}} v_j \, H(\Delta_{p-1,\mv{v},m}),
\]
which in general is not equal to unity.

Still, the moment constraints are not far from being verified. For $\mv{t} \in \Delta_{p-1,\mv{v},m}$ and $j \in \{1, \ldots, p-1\}$ we have
$v_j \le t_j < v_j + 1/m$. Integrating these inequalities over $\mv{t} \in \Delta_{p-1,\mv{v},m}$ with respect to $H$ and summing them over $\mv{v} \in \Vc_{p,m}$ yields
\[
  \int_{\Delta_{p-1}} t_j \, dG_m( \mv{t} ) \le 1 < \int_{\simplexp } t_j \, dG_m( \mv{t} ) + \frac{1}{m}, \qquad j \in \{1, \ldots, p-1\}.
\]
As a consequence, there exist numbers $c_j \in [0, 1)$ such that
\[
  \int_{\simplexp } t_j \, dG_m( \mv{t} ) = 1 - \frac{c_j}{m}, \qquad j \in \{1, \ldots, p-1\}.
\]

\textit{(b) Modification of $G_m$ into a spectral measure $H_m$.} We will modify $G_m$ into a genuine spectral measure $H_m$ by (slightly) increasing the masses at the vertices $e_1, \ldots, e_{p-1}$, where $e_j$ is the $j$th coordinate vector in $\RR^p$. Specifically, we set
\[
  H_m = (1 - a_0) \, G_m + a_1 \, \delta_{e_1} + \cdots + a_{p-1} \, \delta_{e_{p-1}}
\]
for some nonnegative numbers $a_0, \ldots, a_{p-1}$ to be determined by the moment constraints. For $j \in \{1, \ldots, p-1\}$, we must have
\[
  1 = \int_{\Delta_{p-1}} t_j \, dH_m( \mv{t} ) = (1 - a_0) (1 - c_j/m) + a_j.
\]
In addition, the total mass must be equal to
\[
  p = H_m(\Delta_{p-1}) = (1 - a_0) \, p + a_1 + \cdots + a_{p-1}.
\]
Substituting $a_j = 1 - (1 - a_0) (1 - c_j/m)$ into this equation and solving for $a_0$ yields, after some algebra,
\begin{align*}
  a_0 &= \frac{\sum_{i=1}^{p-1} c_i}{m + \sum_{i=1}^{p-1} c_i}, \\
  a_j &= \frac{c_j + \sum_{i=1}^{p-1} c_i}{m + \sum_{i=1}^{p-1} c_i}, \qquad j \in \{1, \ldots, p-1\}.
\end{align*}
This concludes the construction of the spectral measure $H_m$. Note that $0 \le a_j < p/m$ for every $j \in \{0, \ldots, p-1\}$. \bigskip

\noindent\textit{2. Proof of the inequality~\eqref{eq:approx}.}
For $\mv{w}, \mv{t} \in \simplexp$, write 
\[
  f(\mv{w}, \mv{t} ) = \max \{ w_1 t_1, \ldots, w_p t_p \}.
\]
The Pickands dependence function $A_m$ of the spectral measure $H_m$ constructed above is given by
\begin{align*}
  A_m( \mv{w} )  
  &= \int_{\simplexp } f(\mv{w}, \mv{t}) \, H_m(\diff\mv{t}) \\
  &= (1 - a_0) \, \int_{\simplexp} f(\mv{w}, \mv{t} ) \, G_m(\diff\mv{t} ) + a_1 w_1 + \cdots + a_{p-1} w_{p-1}.
\end{align*}

Put $B_m(\mv{w}) = \int f(\mv{w}, \mv{t}) \, G_m(\diff\mv{t})$, the ``Pickands transform'' of $G_m$. Clearly $B_m \ge 0$ and $B_m$ is convex, being a weighted average (over $\mv{t}$) of the convex functions $\mv{w} \mapsto f(\mv{w},\mv{t} )$. As a consequence, $B_m(\mv{w}) \le \max \{ B_m(\mv{e_1}), \ldots, B_m(\mv{e_p}) \}$. Now $B_m(\mv{e_j}) = \int t_j \, G_m(\diff\mv{t})$, which is equal to $1 - c_j/m$ if $j \in \{1, \ldots, p-1\}$ and which is equal to $p - \sum_{i=1}^{p-1} (1 - c_i/m) = 1 + \sum_{i=1}^{p-1} c_i / m = 1 / (1 - a_0)$ if $j = p$. It follows that $B_m(\mv{w}) \le 1 / (1 - a_0)$ for all $\mv{w} \in \simplexp$.

We obtain, on the one hand,
\begin{align*}
  A_m(\mv{w}) 
  &\le B_m(\mv{w}) + a_1 w_1 + \cdots + a_{p-1} w_{p-1} \\
  &\le B_m(\mv{w}) + \max(a_1, \ldots, a_{p-1}) < B_m(\mv{w}) + \frac{p}{m}
\end{align*}
and, on the other hand,
\begin{align*}
  A_m(\mv{w}) 
  \ge (1 - a_0) B_m( \mv{ w } ) 
  \ge B_m(\mv{w}) - \frac{a_0}{1 - a_0}
  &= B_m(\mv{w}) - \frac{1}{m} \sum_{i=1}^{p-1} c_i \\
  &> B_m(\mv{w}) - \frac{p}{m}.
\end{align*}
Therefore,
\[
  |A(\mv{w}) - A_{m}^{\mathrm{P}}(\mv{w})| \le |A(\mv{w}) - B_m(\mv{w})| + |B_m(\mv{w}) - A_m(\mv{w})| < |A(\mv{w}) - B_m(\mv{w})| + \frac{p}{m}.
\]
Furthermore,
\begin{align*}
  |A(\mv{w}) - B_m(\mv{w})|
  &\le \sum_{\mv{v} \in \Vc_{p,m}} \left| \int_{\Delta_{p,\mv{v},m}} f(\mv{w},\mv{t}) \, H(\diff\mv{t}) - \int_{\Delta_{p,\mv{v},m}} f(\mv{w},\mv{t}) \, G_m(\diff\mv{t}) \right| \\
  &= \sum_{\mv{v} \in \Vc_{p,m}} \left| \int_{\Delta_{p,\mv{v},m}} f(\mv{w},\mv{t}) \, H(\diff\mv{t}) - f(\mv{w},\mv{v}) \, H(\Delta_{p,\mv{v},m}) \right| \\
  &= \sum_{\mv{v} \in \Vc_{p,m}} \left| \int_{\Delta_{p,\mv{v},m}} \bigl( f(\mv{w}, \mv{t}) - f(\mv{w}, \mv{v}) \bigr) \, H(\diff\mv{t})  \right| \\
  &\le \sum_{\mv{v} \in \Vc_{p,m}} \int_{\Delta_{p,\mv{v},m}} \bigl| f(\mv{w}, \mv{t}) - f(\mv{w}, \mv{v}) \bigr| \, H(\diff\mv{t}).
\end{align*}
By checking all possible cases one verifies that $|\max(a_1,a_2) - \max(b_1,b_2)| \le \max(|a_1-b_1|, |a_2-b_2|)$ for all real $a_1, a_2, b_1, b_2$. An induction argument yields $|\max(a_1, \ldots, a_k) - \max(b_1, \ldots, b_k)| \le \max(|a_1 - b_1|, \ldots, |a_k - b_k|)$. It follows that $|f(\mv{w}, \mv{t}) - f(\mv{w}, \mv{v})| \le \max(| t_1 - v_1 |, \ldots, | t_p - v_p|)$. As a consequence,
\[
  |A(\mv{w}) - B_m( \mv{w} ) |
  < \sum_{\mv{v} \in \Vc_{p,m}} \int_{\Delta_{p,\mv{v},m}} \frac{p-1}{m} \, H(\diff\mv{t})
  = \frac{(p-1)p}{m}. 
\]
Inequality~\eqref{eq:approx} follows.

\end{document}